\documentclass[11pt]{article}
\usepackage{amsfonts}
\usepackage{mathrsfs}
\usepackage{amssymb}
\usepackage{stmaryrd}
\usepackage{amsmath}
\usepackage{graphicx}
\usepackage{subfig}
\usepackage{color}
\usepackage{ifpdf}
\usepackage{arydshln}
\usepackage{booktabs}
\usepackage{multirow}

\newtheorem{assumption}{Assumption}
\newtheorem{remark}{Remark}
\newtheorem{theorem}{Theorem}
\newtheorem{lemma}{Lemma}

\newtheorem{definition}{Definition}

\newenvironment{proof}{\makebox[7ex][l]{\it Proof:\/}}{\hfill\/ \hfill\/ {\it
Q.E.D.} \vspace{0.5ex}\\}

\begin{document}

\title{\LARGE \bf
State Feedback Control of State-Delayed LPV Systems \\
using Dynamic IQCs
}

\author{Fen Wu\thanks{Email: {\tt fwu@ncsu.edu}, Phone: (919) 515-5268}
\vspace*{0.1in} \\
Department of Mechanical and Aerospace Engineering \\
North Carolina State University \\
Raleigh, NC 27695, USA}

\date{}

\maketitle

\begin{abstract}
This paper develops a new control framework for linear parameter-varying (LPV) systems with time-varying state delays by integrating parameter-dependent Lyapunov functions with integral quadratic constraints (IQCs). A novel delay-dependent state-feedback controller structure is proposed, consisting of a linear state-feedback law augmented with an additional term that captures the delay-dependent dynamics of the plant. Closed-loop stability and $\mathcal{L}_2$-gain performance are analyzed using dynamic IQCs and parameter-dependent quadratic Lyapunov functions, leading to convex synthesis conditions that guarantee performance in terms of parameter-dependent linear matrix inequalities (LMIs). Unlike traditional delay control approaches, the proposed IQC-based framework provides a flexible and systematic methodology for handling delay effects, enabling enhanced control capability, reduced conservatism, and improved closed-loop performance.
\end{abstract}



\section{Introduction}
\label{Sec.Int}

Time-delay systems have received considerable attention over the past decades due to their widespread presence in practical engineering applications, including industrial manufacturing processes \cite{LiuGQYFC.TIE14}, neural networks \cite{LiGY.TCYB11, WanGQ.TNNLS15, KarG.TCYB10}, and fuzzy dynamical systems \cite{QiuFY.TFS09, DonWL.TCYB12, ZhaNS.TC14}. The presence of delays often induces instability and performance degradation, which has motivated the development of various tools for stability analysis and controller design \cite{GuKC.B03, Fri.B14, Bri.B15}.

Early research primarily focused on systems with constant delays, and the theory for stability and stabilization of linear time-invariant (LTI) systems with constant delays is now well established \cite{NicG.B04, Krs.CSM10}. More recently, time-varying delays have attracted increasing attention, as they naturally arise in large-scale complex systems such as networked control systems \cite{LiuGQYFC.TIE14, LiGY.TCYB11, WanGQ.TNNLS15, KarG.TCYB10,YuaW.TCNS17} and underwater vehicles \cite{MaoW.JVC10}. This has led to extensive investigations on the analysis and synthesis of systems with time-varying delays \cite{Fri.B14, Par.TAC99, KaoR.Au07, BriSL.IFAC08}, including the development of various stability and performance criteria \cite{Par.TAC99}. A common approach in these studies is the construction of Lyapunov-Krasovskii functionals (LKFs) to obtain delay-dependent analysis results. Although more sophisticated LKFs often reduce conservatism, identifying the precise sources of conservatism remains challenging, as does designing functionals that achieve a desirable balance between accuracy and computational efficiency.

A further difficulty associated with LKF-based approaches arises in control synthesis. Most synthesis conditions derived from LKFs are inherently non-convex, frequently taking the form of bilinear matrix inequalities (BMIs) due to the coupling between Lyapunov and controller decision matrices \cite{Fri.B14, Bri.B15, BriSL.IFAC08}. To obtain computationally tractable conditions, auxiliary variables or structural constraints are typically introduced, which inevitably introduce additional conservatism. Consequently, a persistent gap exists between analysis and synthesis results, explaining why synthesis conditions are often more conservative despite the rich body of analysis results.
In the context of LPV systems, \cite{WuG.Au01,Wu2001} presents an analysis and state-feedback synthesis framework for LPV systems with parameter-dependent delays. Stability and induced $\mathcal{L}_2$ performance conditions are formulated as LMIs using parameter-dependent Lyapunov–Krasovskii functionals, enabling convex computation.
Subsequently, \cite{BriSL.SCL10} proposes $\delta$-memory-resilient gain-scheduled state-feedback controllers for uncertain LTI/LPV systems, unifying memoryless and exact-memory controllers while explicitly accounting for mismatch between system and controller delays.

Integral quadratic constraints (IQCs), introduced by Megretski and Rantzer \cite{MegR.TAC97}, provide an alternative framework for modeling nonlinearities and uncertainties, including saturation, dead zones, delays, parametric uncertainties, and unmodeled dynamics. IQCs have been successfully applied to stability analysis and stabilization of uncertain dynamical systems \cite{Bri.B15, KaoR.Au07, Kao.TAC12, VeeS.Au14, VeeS.IJRNC14}. For time-delay systems, IQC-based methods offer several advantages over LKF-based approaches by explicitly characterizing the input-output behavior of delayed dynamics. This representation reveals the sources of conservatism more transparently and provides systematic mechanisms for reducing it. Libraries of IQC multipliers for continuous- and discrete-time LTI systems with time-varying delays have been developed \cite{KaoR.Au07, Kao.TAC12}, while alternative formulations include static IQCs, input-output methods \cite{FriS.SCL06}, and quadratic separation techniques \cite{AriG.CDC09}.

Despite these advances, nearly all existing IQC-based results for delayed systems focus primarily on stability and performance analysis, whereas control synthesis within the IQC framework has not been adequately explored except for \cite{VeeS.IJRNC14}. For example,  \cite{PfiS.IJRNC15} presents a robust LPV analysis framework that combines LPV system modeling with uncertainties described by integral quadratic constraints (IQCs)for delayed systems, providing computationally efficient conditions for assessing robust performance. The approach generalizes the nominal LPV bounded real lemma to systems whose state matrices have arbitrary dependence on time-varying parameters, extending applicability beyond prior rational-dependence methods. It also shows significant performance improvement using dynamic IQC with parameter dependent Lyapunov functions. This observation motivates the present work, which exploits the IQC methodology to address the delay control synthesis problem with both stringent performance guarantees and computational efficiency. Recent studies such as \cite{YuaW.TCYB16} introduce a dynamic IQC-based exact-memory control framework for uncertain linear systems with time-varying state delays, embedding the delay operator directly into the controller and yielding fully convex ${\cal H}_\infty$ synthesis conditions with reduced conservatism. Similarly, \cite{YuaW.Au17} develops a dynamic IQC-based exact-memory output-feedback control scheme for linear systems with time-varying input delays, resulting in fully convex LMI conditions for both memoryless and exact-memory controllers. These works highlight the potential of IQC-based frameworks to bridge the longstanding gap between analysis and synthesis.

This paper develops a new control framework for LPV systems with time-varying state delays by integrating parameter-dependent Lyapunov functions with integral quadratic constraints (IQCs). Building upon our prior work on IQC-based analysis and control of uncertain linear systems with state delays and LTI systems with input delays \cite{YuaW.TCYB16,YuaW.Au17}, the proposed approach extends these techniques to the LPV setting, where parameter variations and delay effects interact in a fundamentally nontrivial manner. Specifically, we investigate the state-feedback delay control synthesis problem for LPV systems with time-varying state delays under the IQC framework. A novel delay-dependent controller structure is proposed, consisting of a linear state-feedback law augmented with an additional term that explicitly captures the delay-dependent dynamics of the plant. Under this control scheme, we apply a bounded real lemma  to characterize closed-loop stability and $\mathcal{L}_2$-gain performance using dynamic IQCs and parameter-dependent quadratic Lyapunov functions. Consequently, convex synthesis conditions guaranteeing a prescribed $\mathcal{L}_2$-gain performance are derived in terms of parameter-dependent linear matrix inequalities (LMIs).

Unlike traditional delay analysis and control approaches based on Lyapunov–Krasovskii functionals, the proposed IQC-based framework offers a fundamentally different perspective for handling delay effects.
By leveraging the flexibility of IQCs together with parameter-dependent Lyapunov functions, the method enables enhanced control capability, reduced conservatism, and improved closed-loop performance.
By providing a systematic and powerful approach for both analysis and synthesis of time-delay LPV systems via convex optimization, the resulting framework offers a compelling alternative to conventional delay control methods and paves the way for more advanced designs.

It is worth noting that the proposed results are built upon the well-established robust stability analysis framework using IQCs and dissipation inequalities \cite{Sei.TAC14}. Alternative robust stability conditions for IQC synthesis exist based on the Kalman–Yakubovich–Popov (KYP) lemma (see, e.g., \cite{MegR.TAC97, VeeS.Au14, VeeS.IJRNC14}), which could potentially support complementary delay control synthesis approaches.
However, compared with the KYP-based methods, the dissipation inequality framework offers several distinctive advantages for IQC-based robust analysis and synthesis.
A detailed comparison between these approaches is beyond the scope of this paper.

{\bf Notation}.
$\mathbf{R}$ and $\mathbf{C}$ denote the sets of real and complex numbers respectively.
$\mathbf{R}^{m \times n}$ ($\mathbf{C}^{m \times n}$) represents the set of real (complex) $m \times n$ matrices, while $\mathbf{R}^{n}$ ($\mathbf{C}^n$) denotes the set of real (complex) $n \times 1$ vectors.
The identity matrix of dimension $n$ is denoted by $I_n$.
$\mathbf{S}^{n}$ and $\mathbf{S}^{n}_+$ denote the sets of real symmetric $n\times n$ matrices and positive definite matrices, respectively.
A block diagonal matrix with diagonal blocks $X_1, X_2, \cdots, X_p$ is denoted by ${\it diag}\{X_1,  X_2, \cdots, X_p\}$.
The symbol $\star$ is used in LMIs to indicate entries implied by symmetry.
For two integers $k_1 < k_2$, ${\bf I}[k_1, k_2] := \{k_1, k_1+1, \cdots, k_2\}$.
For $s \in \mathbf{C}$, $\bar{s}$ denotes the complex conjugate of $s$.
For a matrix $M \in \mathbf{C}^{m\times n}$, $M^T$ denotes the  transpose and $M^*$ denotes the complex conjugate transpose.
$\mathbf{RL}_\infty$ denotes proper rational functions with real coefficients and no poles on the imaginary axis, and
$\mathbf{RH}_\infty$ denotes the subset of $\mathbf{RL}_\infty$ consisting of functions analytic in the closed right half-plane.
$\mathbf{RL}_\infty^{m \times n}$ and $\mathbf{RH}_\infty^{m \times n}$ denote the corresponding sets of transfer matrices.
For $G \in \mathbf{RL}_\infty^{m \times n}$, the para-Hermitian conjugate is defined as $G^{\thicksim}(s):=G(-\bar{s})^*$.
For $x \in \mathbf{C}^{n}$, the Euclidean norm is defined by
$\|x\| := (x^* x)^{1/2}$.
$L_{2+}^n$ denotes the space of square-integrable functions $u: [0,\infty) \rightarrow \mathbf{R}^n$ equipped with the norm  $\|u\|_2 := \left(\int^\infty_0 u^T(t) u(t) dt \right)^{1/2} < \infty$.
For $u \in L_{2+}^n$, $u_T$ denotes the truncated signal defined by $u_T(t) = u(t), \forall t \leq T$ and $u_T(t) = 0$ otherwise.
The extended space $L_{2e+}$, consists of functions $u$ such that $u_T \in L_{2+}$ for all $T \geq 0$.

The remainder of the paper is organized as follows.
Section \ref{Sec.Prelim} provides a brief review of IQCs.
Section \ref{Sec.Pre} presents the problem statement and the bounded real lemma for robust stability analysis of state-delayed systems using dynamic IQCs.
The corresponding delay control synthesis conditions are derived in Section \ref{Sec.Syn}.
Numerical studies illustrating the effectiveness and advantages of the proposed control design approach are presented in Section \ref{Sec.Ex}.
Finally, conclusions are drawn in Section \ref{Sec.Con}.

\section{Preliminaries}
\label{Sec.Prelim}

In this section, we briefly review several fundamental concepts related to integral quadratic constraints (IQCs), which form the basis for the subsequent developments of this paper.

\begin{definition}[\cite{Sei.TAC14}]
\label{Def.IQC}
Let $\Pi \in \mathbf{RL}_\infty^{(m_1+m_2) \times (m_1+m_2)}$ be a proper rational function, referred to as a multiplier, with factorization $\Pi = \Psi^{\thicksim} W \Psi$, where $W \in \mathbf{S}^{n_z}$ and $\Psi \in \mathbf{RH}_\infty^{n_z \times (m_1+m_2)}$.
Two signals $v\in L_{2e+}^{n_v}$ and $w\in L_{2e+}^{n_w}$ are said to satisfy the IQC defined by $\Pi$ and $(\Psi, W)$ is called a hard IQC factorization of $\Pi$, if the inequality  holds
\begin{align}
\label{Def-hardIQC-ine}
\int^T_0 z^T(t) W z(t) dt \geq 0
\end{align}
holds for all $T \geq 0$, where
$z\in\mathbf{R}^{n_z}$ denotes the filtered output of $\Psi$ driven by inputs $(v,w)$ with zero initial conditions, i.e.,  $z = \Psi\begin{bmatrix} v\\ w \end{bmatrix}$.
Moreover, a bounded causal operator $\mathcal{S}: L_{2e+}^{n_v} \rightarrow L_{2e+}^{n_w}$ satisfies the IQC defined by $\Pi$ if condition (\ref{Def-hardIQC-ine}) holds all admissible signals $v$, with $w = \mathcal{S}(v)$, and for all $T\geq 0$.
\end{definition}

The factorization $\Pi = \Psi^{\thicksim} W \Psi$ is generally not unique but can be computed using standard state-space techniques \cite{SchW.B00}.
Furthermore, a broad class of multipliers admits hard factorizations \cite{MegR.TAC97}.
Additional discussions on hard IQCs can be found in \cite{Sei.TAC14, PfiS.IJRNC15}.
The concept of hard IQCs, together with the following spectral factorization definition and lemma \cite{Sei.TAC14}, plays a central role in the dissipation inequality framework adopted in this work.

\begin{definition}
\label{Def-Jspectral}
A pair $(\Psi, W)$ is called a $J_{m_1,m_2}$-spectral factorization of $\Pi = \Pi^{\thicksim} \in \mathbf{RL}_\infty^{(m_1+m_2) \times (m_1+m_2)}$ if
\[
\Pi = \Psi^\thicksim W \Psi, \qquad W = \begin{bmatrix} X_{1} & 0 \\ 0 & -X_{2} \end{bmatrix},
\]
where $\Psi, \Psi^{-1} \in \mathbf{RH}_\infty^{(m_1+m_2) \times (m_1+m_2)}$, and $X_1\in \mathbf{S}_+^{m_1}$, $X_2\in \mathbf{S}_+^{m_2}$.
\end{definition}

Under a $J_{m_1,m_2}$-spectral factorization, $\Psi$ is square, stable, and minimum phase.

\begin{lemma}[\cite{Sei.TAC14}]
\label{Lemma-Jfactor}
Let $\Pi = \Pi^{\thicksim} \in \mathbf{RL}_\infty^{(m_1+m_2) \times (m_1+m_2)}$ be partitioned as
\[
\Pi = \begin{bmatrix} \Pi_{11} & \Pi_{12} \\ \Pi_{12}^{\thicksim} & \Pi_{22}\end{bmatrix},
\]
where $\Pi_{11} \in \mathbf{RL}_\infty^{m_1 \times m_1}$ and $\Pi_{22} \in \mathbf{RL}_\infty^{m_2 \times m_2}$.
If $\Pi_{11}(j \omega) > 0$ and $\Pi_{22}(j \omega) < 0$ for all $\omega \in \mathbf{R} \cup \{\infty\}$,
then $\Pi$ admits a $J_{m_1,m_2}$-spectral factorization $(\Psi,W)$, also constitutes a hard factorization.
\end{lemma}

IQCs provide a powerful and flexible tool for modeling a wide range of nonlinearities and uncertainties, and have been extensively applied to robust stability and performance analysis of dynamical systems \cite{MegR.TAC97, KaoR.Au07, Kao.TAC12, Sei.TAC14, PfiS.IJRNC15}.
In particular, libraries of IQC multipliers for time-varying delays are available in \cite{KaoR.Au07, PfiS.Au15} for continuous-time systems and \cite{Kao.TAC12} for discrete-time systems.
Due to space limitations, readers are referred to these references for further details.

\section{Problem Formulation}
\label{Sec.Pre}

Consider a class of linear parameter-varying (LPV) systems with time-varying state delays described by
\begin{align}
\label{Plant}
\begin{aligned}
    \dot{x}_p &= A_{p}(\rho) x_p + A_{d}(\rho) \mathcal{D}_{\bar{\tau},r}(x_p) + B_{p1}(\rho) d +  B_{p2}(\rho) u \\
    e &= C_{p1}(\rho) x_p + C_{d1}(\rho) \mathcal{D}_{\bar{\tau},r}(x_p) + D_{p11}(\rho) d +  D_{p12}(\rho) u
\end{aligned},
\end{align}
where $x_p \in \mathbf{R}^{n_x}$ denotes the plant state, $d \in \mathbf{R}^{n_d}$ the disturbance input, $u \in \mathbf{R}^{n_u}$ the control input, and $e \in \mathbf{R}^{n_e}$ the controlled output.
The operator $\mathcal{D}_{\bar{\tau},r}(x_p) := x_p(t - \tau(t))$ represents the time-varying state delay, with the delay function $\tau(t)$ belonging to
\begin{align}
\label{Tauset}
\tau \in \mathcal{T}_{\bar{\tau},r} := \{\tau:\mathbf{R}_+ \rightarrow [0,\bar{\tau}], |\dot{\tau}(t)|\leq r\},
\end{align}
where $\bar{\tau}$ and $r$ denote the maximum delay and its variation bound, respectively.
For notational simplicity, we assume $x_p(t)=0$ for all $t \in [-\bar{\tau},0]$.
For delay-dependent control, the delay $\tau(t)$ is assumed unknown a priori but measurable in real time.
All state-space matrices are continuous functions of the scheduling parameter
$\rho \in {\cal P} \subset {\bf R}^s$, where $\mathcal{P}$ is compact.
The parameter variation rate satisfies
\[
{\cal V} = \left\{ \nu: \underline{\nu}_k \leq \dot{\rho}_k \leq \bar{\nu}_k, k \in {\bf I}[1,s] \right\},
\]
with $\mathcal{V}$ being a convex polytope containing the origin.

To facilitate stability analysis and controller synthesis, we introduce a model transformation that separates the state-delay nonlinearity from the nominal dynamics. Specifically, system (\ref{Plant}) can be rewritten as the interconnected system
\begin{align}
\label{Plant2}
\begin{aligned}
    \mathcal{P}_{nom}: \quad \begin{bmatrix}
        \dot{x}_p \\
        e
    \end{bmatrix} &= \begin{bmatrix}
        A_{p}(\rho)+A_{d}(\rho) & -A_{d}(\rho) & B_{p1}(\rho) & B_{p2}(\rho) \\
        C_{p1}(\rho)+C_{d1}(\rho) & -C_{d1}(\rho) & D_{p11}(\rho) & D_{p12}(\rho)
    \end{bmatrix} \begin{bmatrix}
        x_p \\
        w \\
        d \\
        u
    \end{bmatrix}, \\
    w &= \mathcal{S}_{\bar{\tau},r}(x_p) := x_p - \mathcal{D}_{\bar{\tau},r}(x_p).
\end{aligned}
\end{align}
This representation consists of a delay-free nominal subsystem $\mathcal{P}{nom}$ interconnected with the delay-induced operator $\mathcal{S}{\bar{\tau},r}(\cdot)$.
Such LFT-based representations are standard in robust control theory \cite{ZhouDG.B96, LiGY.TCYB11}.
The input–output behavior of $\mathcal{S}_{\bar{\tau},r}$ will be characterized using dynamic IQCs.


This paper addresses the $\mathcal{H}_\infty$ delay control problem for LPV system (\ref{Plant2}).
The objective is to design a state-feedback controller ensuring $\| e \|_2 < \gamma \| d \|_2$
under zero initial conditions, for all admissible $(\rho,\dot{\rho}) \in \mathcal{P}\times\mathcal{V}$ and $\tau(t) \in \mathcal{T}_{\bar{\tau},r}$.
More specifically, we will solve the $\mathcal{H}_\infty$ synthesis problem for LPV system (\ref{Plant2}) by employing dynamic IQCs to characterize the input-output behavior of the time-varying delay nonlinearity.

For this purpose, we impose the following assumptions.
\begin{assumption}
\label{Ass1}
The pair $(A_p(\rho)+A_d(\rho),B_{p2}(\rho))$ is parametrically stabilizable.
\end{assumption}

\begin{assumption}
\label{Ass2}
The operator $\mathcal{S}_{\bar{\tau},r}$ satisfies IQCs defined by multipliers $\{\Pi_k\}_{k=1}^{N_\lambda} \in \mathbf{RL}_\infty^{2n_x \times 2n_x}$, where the multipliers $\{\Pi_k\}_{k=1}^{N_\lambda}$ with strict definiteness properties ensuring the existence of $J$-spectral factorizations $(\Psi_k, W_k)$. Specifically,
$(\Psi_k, W_k)$ is in the form of $\Psi_k = \begin{bmatrix} \Psi_{11,k} & \Psi_{12,k} \\ 0 & I_{n_x} \end{bmatrix}\in\mathbf{RH}_\infty^{(n_x + n_x) \times (n_x + n_x)}$ and $W_k = \begin{bmatrix}X_k & 0 \\ 0 & -X_k \end{bmatrix}\in\mathbf{R}^{(n_x + n_x) \times (n_x + n_x)}$, where $X_k\in\mathbf{S}^{n_x}_+$.
\end{assumption}

Assumption \ref{Ass1} is standard and guarantees stabilizability of the delay-free system.
Assumption \ref{Ass2} is nonrestrictive and adopted for analytical convenience.

For each multiplier, the associated IQC filter is realized as
\begin{align}
\label{IQCsystem}
\begin{bmatrix}
    \dot{x}_{\psi} \\
    z_{k}
\end{bmatrix} &= \begin{bmatrix}
    A_{\psi} & B_{\psi1} & B_{\psi2} \\
    C_{\psi,k} & D_{\psi1,k} & D_{\psi2,k}
\end{bmatrix} \begin{bmatrix}
    x_{\psi} \\
    x_p \\
    w
\end{bmatrix}, \quad k \in {\bf I}[1,N_\lambda]
\end{align}
where $x_\psi \in \mathbf{R}^{n_\psi}$ denotes the state vector of the IQC-induced operator $\{\Psi_k\}_{k=1}^{N_\lambda}$, with zero initial condition $x_\psi(0) = 0$.
The signals $z_k \in \mathbf{R}^{n_z}$ with $n_z = 2 n_x$ for all $k \in {\bf I}[1,N_\lambda]$, represent the corresponding operator outputs.
Under Assumption \ref{Ass2}, the associated output matrices possess the following structure:
\begin{align*}
C_{\psi,k} &= \begin{bmatrix}
    \bar{C}_{\psi,k} \\
    0
\end{bmatrix}, \quad
D_{\psi1,k} = \begin{bmatrix}
    \bar{D}_{\psi1,k} \\
    0
\end{bmatrix}, \quad
D_{\psi2,k} = \begin{bmatrix}
    \bar{D}_{\psi2,k} \\
    I_{n_x}
\end{bmatrix}, \quad k \in {\bf I}[1,N_\lambda].
\end{align*}

For state-feedback control, we assume that all relevant state information, including the plant state $x_p$ and the IQC-induced system state $x_\psi$ in (\ref{IQCsystem}), is available for feedback.
It should be noted that the IQC-induced LTI system (\ref{IQCsystem}) is introduced solely for robustness analysis and controller synthesis.
Once the realization of (\ref{IQCsystem}) is fixed, the corresponding state $x_\psi$ can be computed online, since the input signals $x_p$ and $w = \mathcal{S}_{\bar{\tau},r}(x_p)$ are available in real time.
Therefore, we propose the delay-dependent state-feedback controller in LPV form:
\begin{align}
\label{Controller}
u = F_c(\rho) \begin{bmatrix}
    x_p \\
    x_{\psi}
\end{bmatrix} + H_c(\rho) \mathcal{S}_{\bar{\tau},r}(x_p),
\end{align}
where $F_c(\rho) \in \mathbf{R}^{n_u \times (n_x+n_\psi)}$ and $H_c(\rho) \in \mathbf{R}^{n_u \times n_x}$ are controller gain matrices to be designed.
This controller explicitly incorporates delay-dependent dynamics.
In particular, the proposed delay-compensating mechanism utilizes the past state information of the plant $x_p$ over the interval $[t-\bar{\tau}, t]$.
This structure enables the controller gains to be scheduled online based on real-time measurements of the delay parameter $\tau(t)$ and parameter $\rho$.

The resulting closed-loop system by interconnecting the controlled plant (\ref{Plant2}), the IQC-induced system (\ref{IQCsystem}) and the controller (\ref{Controller}) is given by
\begin{align}
\label{CL}
\begin{aligned}
    \begin{bmatrix}
    \dot{x}_{cl} \\
    z_k \\
    e
    \end{bmatrix} &= \begin{bmatrix}
    A_{cl}(\rho) & B_{cl1}(\rho) & B_{cl2}(\rho) \\
    C_{cl1,k} & D_{cl11,k} & D_{cl12,k} \\
    C_{cl2}(\rho) & D_{cl21}(\rho) & D_{cl22}(\rho)
    \end{bmatrix} \begin{bmatrix}
    x_{cl} \\
    w \\
    d
    \end{bmatrix}, \\
    w &= \mathcal{S}_{\bar{\tau},r}(x_p),
\end{aligned}
\end{align}
where $x_{cl}:=[x_p^T \ \ x_{\psi}^T]^T$ and
\begin{align}
\label{OF-CloseLoop}
\begin{bmatrix}
A_{cl}(\rho) & B_{cl1}(\rho) & B_{cl2}(\rho) \\
C_{cl1,k} & D_{cl11,k} & D_{cl12,k} \\
C_{cl2}(\rho) & D_{cl21}(\rho) & D_{cl22}(\rho)
\end{bmatrix}
&= \left[ \begin{array}{cc:cc}
A_p(\rho)+A_d(\rho) & 0 & -A_d(\rho) & B_{p1}(\rho) \\
B_{\psi1} & A_{\psi} & B_{\psi2} & 0 \\ \hdashline
D_{\psi1,k} & C_{\psi,k} & D_{\psi2,k} & 0 \\
C_{p1}(\rho)+C_{d1}(\rho) & 0 & -C_{d1}(\rho) & D_{p11}(\rho)
\end{array} \right] \nonumber \\
& \hspace*{1.0in} + \left[ \begin{array}{c}
    B_{p2}(\rho) \\
    0 \\ \hdashline
    0 \\
    D_{p12}
\end{array} \right] \begin{bmatrix}
    F_c(\rho) & H_c(\rho)
\end{bmatrix} \begin{bmatrix}
    I & 0 & 0 \\
    0 & I & 0
\end{bmatrix} \\
& \hspace*{-1.5in} := \begin{bmatrix}
A_{aug}(\rho) & B_{aug0}(\rho) & B_{aug1}(\rho) \\
C_{aug0} & D_{aug00} & 0 \\
C_{aug1}(\rho) & D_{aug10}(\rho) & D_{aug11}(\rho)
\end{bmatrix} + \begin{bmatrix}
B_{aug2}(\rho) \\
0 \\
D_{aug12}(\rho)
\end{bmatrix} \begin{bmatrix}
    F_c(\rho) & H_c(\rho)
\end{bmatrix} \begin{bmatrix}
    I & 0 & 0 \\
    0 & I & 0
\end{bmatrix}. \nonumber
\end{align}
In particular, for all $k \in {\bf I}[1,N_\lambda]$, we have
\begin{align*}
C_{cl1,k} &:= \begin{bmatrix}
    \bar{C}_{cl1,k} \\
    0
\end{bmatrix} = \begin{bmatrix}
    \bar{D}_{\psi1,k} & \bar{C}_{\psi,k} \\
    0 & 0
\end{bmatrix}, \\
\begin{bmatrix}
    D_{cl11,k} & D_{cl12,k}
\end{bmatrix} &:= \begin{bmatrix}
    \bar{D}_{cl11,k} & \bar{D}_{cl12,k}
\\
    I_{n_x} & 0
\end{bmatrix} = \begin{bmatrix}
    \bar{D}_{\psi2,k} & 0 \\
    I_{n_x} & 0
\end{bmatrix}.
\end{align*}

We now recall the bounded real lemma for robust  stability and $\mathcal{L}_2$ gain performance under IQCs \cite{PfiS.IJRNC15}.

\begin{theorem}
\label{Ana-Thm}
Consider the closed-loop system (\ref{CL}).
If there exist parameter-dependent matrices $P(\rho) \in \mathbf{S}_+^{n_x+n_\psi}$, $X_k(\rho) \in \mathbf{S}_+^{n_x}$ for all $k \in {\bf I}[1,N_\lambda]$, and scalar $\gamma > 0$ satisfying the following LMI
\begin{align}
\label{Ana-Thm-LMI}
\begin{aligned}
    &\begin{bmatrix}
        \begin{matrix} {\it He} \{P(\rho) A_{cl}(\rho)\} \\ + \sum_{i=1}^s \left\{ \underline{\nu}_i, \bar{\nu}_i \right\} \frac{\partial P}{\partial \rho_i} \end{matrix} & \star & \star \\
        B_{cl1}^T(\rho) P(\rho) & -\sum^{N_\lambda}_{k=1} X_k(\rho) & \star \\
        B_{cl2}^T P & 0 & -\gamma I_{n_d}
    \end{bmatrix} + \sum^{N_\lambda}_{k=1} \begin{bmatrix}
        \bar{C}_{cl1,k}^T \\
        \bar{D}_{cl11,k}^T \\
        \bar{D}_{cl12,k}^T
    \end{bmatrix} X_k(\rho) \begin{bmatrix}
        \bar{C}_{cl1,k} & \bar{D}_{cl11,k} & \bar{D}_{cl12,k}
    \end{bmatrix}  \\
    &\hspace*{0in} + \frac{1}{\gamma}
    \begin{bmatrix}
        C_{cl2}^T(\rho) \\
        D_{cl21}^T(\rho) \\
        D_{cl22}^T(\rho)
    \end{bmatrix} \begin{bmatrix}
        C_{cl2}(\rho) & D_{cl21}(\rho) & D_{cl22}(\rho)
    \end{bmatrix} < 0
\end{aligned}
\end{align}
for all $(\rho, \dot{\rho}) \in {\cal P} \times {\cal V}$,
then the closed-loop system (\ref{CL}) with time-varying state-delay is stable with $\mathcal{L}_2$ gain less than $\gamma$.
\end{theorem}

\section{Delayed-Dependent Control Synthesis}
\label{Sec.Syn}

Based on the analysis results in Section \ref{Sec.Pre}, we now address the $\mathcal{H}_\infty$ synthesis problem for the delay-dependent controller (\ref{Controller}). The synthesis conditions are formulated in terms of parameter-dependent LMIs, as stated in the following theorem.

\begin{theorem}
\label{Syn-Thm}
Consider the state-delayed LPV plant (\ref{Plant}). Suppose there exist positive-definite matrices
$R(\rho) \in \mathbf{S}_+^{n_x+n_\psi}$, $\hat{X}_k(\rho) \in \mathbf{S}_+^{n_x}$ for all $k\in{\bf I}[1,N_\lambda]$, a rectangular matrix $\hat{X}(\rho) \in \mathbf{R}^{n_x \times n_x}$,
and a scalar $\gamma > 0$ such that, for all $\rho \in {\cal P}$ the following LMIs hold.
\begin{align}
& {\cal N}_R^T(\rho) \begin{bmatrix}
    \begin{matrix} {\it He} \left\{ \begin{bmatrix}
        A_{p}(\rho) +A_{d}(\rho) & 0 \\
        B_{\psi1} & A_{\psi}
    \end{bmatrix} R(\rho) \right\} \\ - \left\{ \underline{\nu}_i, \bar{\nu}_i \right\} \frac{\partial R}{\partial \rho} \end{matrix} & \star & \star & \star & \star \\
    \hat{X}^T(\rho) \begin{bmatrix}
        -A_{d}^T(\rho) & B_{\psi2}^T
    \end{bmatrix} & -\displaystyle{\sum^{N_\lambda}_{k=1}} (\hat{X}(\rho) + \hat{X}^T(\rho) - \hat{X}_k(\rho)) & \star & \star & \star \\
    \begin{bmatrix}
        B_{p1}^T(\rho) & 0
    \end{bmatrix} & 0 & -\gamma I_{n_d} & \star & \star \\
    \hat{\Upsilon}_{41}(\rho) & \hat{\Upsilon}_{42}(\rho) & 0 & -\hat{\Lambda}(\rho) & \star \\
    \begin{bmatrix}
        C_{p1}(\rho) + C_{d1}(\rho) & 0
    \end{bmatrix} R(\rho) & -C_{d1}(\rho) \hat{X}(\rho) & D_{p11}(\rho) & 0 & -\gamma I_{n_e}
\end{bmatrix} {\cal N}_R(\rho) \nonumber \\
&= {\cal N}_R^T(\rho) \begin{bmatrix}
    {\it He} \left\{ A_{aug}(\rho) R(\rho) \right\} - \left\{ \underline{\nu}_i, \bar{\nu}_i \right\} \frac{\partial R}{\partial \rho} & \star & \star & \star & \star \\
    \hat{X}^T(\rho) B_{aug0}^T(\rho) & -\displaystyle{\sum^{N_\lambda}_{k=1}} (\hat{X}(\rho) + \hat{X}^T(\rho) - \hat{X}_k(\rho)) & \star & \star & \star \\
    B_{aug1}^T(\rho) & 0 & -\gamma I_{n_d} & \star & \star \\
    \hat{\Upsilon}_{41}(\rho) & \hat{\Upsilon}_{42}(\rho) & 0 & -\hat{\Lambda}(\rho) & \star \\
    C_{aug1}(\rho) R(\rho) & D_{aug10}(\rho) \hat{X}(\rho) & D_{aug11}(\rho) & 0 & -\gamma I_{n_e}
\end{bmatrix} {\cal N}_R(\rho) \nonumber \\
&< 0, \label{Syn-Thm-LMI1} \\
& \begin{bmatrix}
-\gamma I & \star & \star \\
0 & -\hat{\Lambda} & 0 \\
D_{p11}(\rho) & 0 & -\gamma I
\end{bmatrix} < 0, \label{Syn-Thm-LMI2}
\end{align}
where ${\cal N}_R(\rho) = {\rm Ker} [B_{p2}^T(\rho) \ \ 0 \ \ 0 \ \ 0 \ \ 0 \ \ D_{p12}^T(\rho)] = {\rm Ker} [B_{aug2}^T(\rho) \ \ 0 \ \ 0 \ \ 0 \ \ D_{aug12}^T(\rho)]$ and
\begin{align*}
\hat{\Upsilon}_{41}(\rho) &= \begin{bmatrix}
    \begin{bmatrix}
        \bar{D}_{\psi1,1} & \bar{C}_{\psi1,1}
    \end{bmatrix} R(\rho) \\
    \vdots \\
    \begin{bmatrix}
        \bar{D}_{\psi1,N_\lambda} & \bar{C}_{\psi1,N_\lambda}
    \end{bmatrix} R(\rho)
\end{bmatrix} = C_{aug0} R(\rho), \qquad
\hat{\Upsilon}_{42}(\rho) = \begin{bmatrix}
    \bar{D}_{\psi2,1} \hat{X}(\rho) \\
    \vdots \\
    \bar{D}_{\psi2,N_\lambda} \hat{X}(\rho)
\end{bmatrix} = D_{aug00} \hat{X}(\rho), \\
\hat{\Lambda}(\rho) &= {\it diag}\{\hat{X}_1(\rho), \cdots, \hat{X}_{N_\lambda}(\rho)\}.
\end{align*}
Then, the LPV system (\ref{Plant}) with time-varying state delay is stabilized by the state-feedback controller (\ref{Controller})
and achieves an $\mathcal{L}_2$ gain less than $\gamma$ for all parameter trajectories $(\rho, \dot{\rho}) \in {\cal P} \times {\cal V}$ and all admissible delay $\tau\in \mathcal{T}_{\bar{\tau},r}$.
\end{theorem}

\begin{proof}
From Theorem \ref{Ana-Thm}, the closed-loop system  (\ref{CL}) is stable and achieves an $\mathcal{L}_2$ gain less than $\gamma$ if condition (\ref{Ana-Thm-LMI}) is satisfied.
This directly implies that the state-delayed LPV system (\ref{Plant}) is stabilized with the prescribed performance level.

Under Assumption \ref{Ass2}, substituting the IQC scaling matrices $W_k, k \in {\bf I}[1,N_\lambda]$ into condition (\ref{Ana-Thm-LMI}), and applying the Schur complement
yields condition
\begin{align}
\label{Prf-Syn-LMI}
\begin{bmatrix}
        {\it He}\{P A_{cl}\} + \dot{P} & \star & \star & \star & \star \\
        B_{cl1}^T P & -\sum^{N_\lambda}_{k=1}X_k & \star & \star & \star \\
        B_{cl2}^T P & 0 & -\gamma I_{n_d} & \star & \star \\
        \Upsilon_{41} & \Upsilon_{42} & \Upsilon_{43} & -\Lambda & \star \\
        C_{cl2} & D_{cl21} & D_{cl22} & 0 & -\gamma I_{n_e}
\end{bmatrix} < 0,
\end{align}
where
\begin{align*}
\Upsilon_{41} &= \begin{bmatrix}
    \bar{C}_{cl1,1} \\
    \vdots \\
    \bar{C}_{cl1,N_\lambda}
\end{bmatrix}, \quad
\Upsilon_{42} = \begin{bmatrix}
    \bar{D}_{cl11,1} \\
    \vdots \\
    \bar{D}_{cl11,N_\lambda}
\end{bmatrix}, \quad
\Upsilon_{43} = \begin{bmatrix}
    \bar{D}_{cl12,1} \\
    \vdots \\
    \bar{D}_{cl12,N_\lambda}
\end{bmatrix}, \\
\Lambda &= {\it diag} \left\{ X_1^{-1},\cdots, X_{N_\lambda}^{-1} \right\}.
\end{align*}
Then, define $R = P^{-1}$ and $\hat{X}_k = X_k^{-1}$ for all $k \in {\bf I}[1,N_\lambda]$. Using the inequality  $-\hat{X}_k^{-1} \leq - \hat{X}^{-T}(\hat{X}^T + \hat{X} - \hat{X}_k)\hat{X}^{-1}$, which holds for any non-singular matrix $\hat{X} \in \mathbf{R}^{n_x \times n_x}$, a sufficient condition for (\ref{Prf-Syn-LMI}) is given by
\begin{align}
\label{SF-Prf-Syn-LMI2}
\begin{aligned}
    \begin{bmatrix}
        {\it He}\{R^{-1} A_{cl} \} + \dot{R^{-1}} & \star & \star & \star & \star \\
        B_{cl1}^T R^{-1} & -\hat{X}^{-T}\sum^{N_\lambda}_{k=1} (\hat{X} + \hat{X}^T - \hat{X}_k) \hat{X}^{-1} & \star & \star & \star \\
        B_{cl2}^T R^{-1} & 0 & -\gamma I_{n_d} & \star & \star \\
        \Upsilon_{41} & \Upsilon_{42} & \Upsilon_{43} & -\Lambda & \star \\
        C_{cl2} & D_{cl21} & D_{cl22} & 0 & -\gamma I_{n_e}
    \end{bmatrix} < 0.
\end{aligned}
\end{align}
Rewriting this condition in terms of the augmented system matrices leads to
the following inequality
\begin{align}
& \begin{bmatrix}
        {\it He}\{A_{aug} R\} - \dot{R} & \star & \star & \star & \star \\
        \hat{X}^T B_{aug0}^T & -\sum^{N_\lambda}_{k=1} (\hat{X} + \hat{X}^T - \hat{X}_k) & \star & \star & \star \\
        B_{aug1}^T & 0 & -\gamma I_{n_d} & \star & \star \\
        C_{aug0} R & D_{aug00} \hat{X} & 0 & -\hat{\Lambda} & \star \\
        C_{aug1} R & D_{aug10} \hat{X} & D_{aug11} & 0 & -\gamma I_{n_e}
\end{bmatrix} \nonumber \\
& \hspace*{0.5in} + \begin{bmatrix}
    B_{aug2} \\
    0 \\
    0 \\
    0 \\
    D_{aug12}
\end{bmatrix} \begin{bmatrix}
    F_c R & H_c \hat{X}
\end{bmatrix} \begin{bmatrix}
    I & 0 & 0 & 0 & 0 \\
    0 & I & 0 & 0 & 0
\end{bmatrix} + \begin{bmatrix}
    I & 0 \\
    0 & I \\
    0 & 0 \\
    0 & 0 \\
    0 & 0
 \end{bmatrix} \begin{bmatrix}
    R F_c^T \\
    \hat{X}^T H_c^T
    \end{bmatrix} \begin{bmatrix}
    B_{aug2}^T & 0 & 0 & 0 & D_{aug12}^T
    \end{bmatrix} < 0. \label{temp}
\end{align}
Finally, applying the Elimination Lemma \cite{BoyEFB.B94} yields
conditions (\ref{Syn-Thm-LMI1}) and (\ref{Syn-Thm-LMI2}), completing the proof.
\end{proof}

Conditions (\ref{Syn-Thm-LMI1}) and (\ref{Syn-Thm-LMI2}) are parameter-dependent LMIs.
Since $R(\rho)$, $\hat{X}(\rho)$, and $\hat{X}_k(\rho)$ are parameter-dependent matrix
functions, they are approximated using basis function expansions:
\begin{align*}
R(\rho) &= \sum_{i=1}^{n_f} f_i(\rho) R_i, \\
\hat{X}(\rho) &= \sum_{i=1}^{n_g} g_i(\rho) \hat{X}_i, \qquad \hat{X}_k(\rho) = \sum_{i=1}^{n_g} g_i(\rho) \hat{X}_{k,i}
\end{align*}
A finite set of gridding points over the parameter space ${\cal P}$ is then selected.
The control design problem then reduces to the convex optimization problem
\begin{equation}
\label{Opt}
\begin{aligned}
    &\min_{R_i, \hat{X}_i, \hat{X}_{k,i}, \ \forall k \in {\bf I}[1,N_\lambda]} \qquad \gamma \\
    &\hspace*{0.6in} \mbox{s.t.} \qquad \mbox{(\ref{Syn-Thm-LMI1})-(\ref{Syn-Thm-LMI2}).}
\end{aligned}
\end{equation}
Solving (\ref{Opt}) yields the matrix functions $R(\rho)$, $\hat{X}(\rho)$,
and $\hat{X}_k(\rho)$. The controller variables
$\hat{F}_c(\rho) = F_c(\rho) R(\rho)$ and
$\hat{H}_c(\rho) = H_c(\rho) \hat{X}(\rho)$
can then be computed pointwise by enforcing feasibility of (\ref{temp}).
Consequently, the delay-dependent controller gains are recovered as
\[
F_c(\rho) = \hat{F}_c(\rho) R^{-1}(\rho), \qquad H_c(\rho) = \hat{H}_c(\rho) \hat{X}^{-1}(\rho).
\]


\begin{remark}
The derivation of Theorem \ref{Syn-Thm} involves a relaxation from (\ref{Prf-Syn-LMI}) to (\ref{SF-Prf-Syn-LMI2}) using the inequality
$-\hat{X}_k^{-1} \leq - \hat{X}^{-T}(\hat{X}^T + \hat{X} - \hat{X}_k)\hat{X}^{-1}$.
This relaxation may introduce conservatism when multiple IQC multipliers are employed.
However, for the single-multiplier case ($N_\lambda = 1$), the synthesis condition (\ref{Syn-Thm-LMI1}) is equivalent to the analysis condition (\ref{Prf-Syn-LMI}).
Importantly, in both cases, the synthesis conditions are
convex and can be efficiently solved using semidefinite programming \cite{BoyEFB.B94}.
\end{remark}

\section{Numerical Example}
\label{Sec.Ex}

In this section, a numerical example is presented to illustrate the design procedure and performance of the proposed delay-dependent LPV controller. Consider the LPV state-delayed system
\begin{align}
\label{Ex.Plant}
\begin{aligned}
    \dot{x}_p(t) &= \begin{bmatrix}
        0 & 1+\phi \rho(t) \\
        -2 & -3+\sigma\rho(t)
    \end{bmatrix} x_p(t) + \begin{bmatrix}
        \phi \rho(t) & 0.1 \\
        -0.2+\sigma\rho(t) &-0.3
    \end{bmatrix} x_p(t-\tau(t)) \\
    &\hspace*{0.4in} + \begin{bmatrix}
        0.2 \\
        0.2
    \end{bmatrix}d(t) + \begin{bmatrix}
        \phi\rho(t) \\
        0.1 + \sigma\rho(t)
    \end{bmatrix} u(t), \\
    e(t) &= \begin{bmatrix}
        0 & 10 \\
        0 & 0
    \end{bmatrix} x_p(t) + \begin{bmatrix}
        0 \\
        0.1
    \end{bmatrix} u(t),
\end{aligned}
\end{align}
where $\phi = 0.2$, $\sigma = 0.1$, and the time-varying parameter $\rho \in [-1, 1]$.
The delay $\tau(t)$ satisfies $\tau\in \mathcal{T}_{\bar{\tau},r}$. This system has also been studied in \cite{WuG.Au01, ZhaG.MCCA05, BriSL.SCL10}.

For the IQC-based delay-dependent LPV control synthesis, two dynamic IQC multipliers from \cite{KaoR.Au07} could be used to characterize the delay operator $\mathcal{S}_{\bar{\tau},r}(x_p)$:
\begin{align}
\label{Ex.Pi12}
\begin{aligned}
    \Pi_1(s) &= \begin{bmatrix}
        |\phi(s)|^2 X_1 & 0 \\
        0 & -X_1
    \end{bmatrix}, \quad
    \Pi_2(s) = \begin{bmatrix}
        |\varphi(s)|^2 X_2 & 0 \\
        0 & -X_2
    \end{bmatrix}
\end{aligned}
\end{align}
where $X_1,X_2 \in \mathbf{S}_+^{2}$,
\[
\phi(s) = k_1 \left(\frac{\bar{\tau}^2 s^2 + c_1 \bar{\tau} s}{\bar{\tau}^2 s^2 + a_1 \bar{\tau} s + k_1 c_1}\right) + \epsilon, \qquad \varphi(s) = k_2 \left(\frac{\bar{\tau}^2 s^2 + c_2 \bar{\tau} s}{\bar{\tau}^2 s^2 + a_2 \bar{\tau} s + b_2}\right) +\delta,
\]
with $k_1 = 1 + \frac{1}{\sqrt{1-r}}, a_1 = \sqrt{2 k_1 c_1}$, $c_1 < 2 k_1$, $k_2 = \sqrt{\frac{8}{2-r}}, a_2 = \sqrt{6.5 + 2 b_2}$, $b_2 = \sqrt{50}$, $c_2 = \sqrt{12.5}$, and
arbitrarily small $\epsilon, \delta > 0$.
In this example, $c_1 = 1, \epsilon = 10^{-7}$, and $\delta = 0.0001$.

Applying the IQC factorization methods from \cite{Sei.TAC14, PfiS.IJRNC15} yields $J$-spectral factorizations of $\Pi_1$ and $\Pi_2$ in (\ref{Ex.Pi12}):
\begin{align}
\label{Ex.Psi12}
\begin{aligned}
    \Psi_1(s) &= \begin{bmatrix}
        \left( \frac{\frac{k_1(c_1-a_1)}{\bar{\tau}} s - k_1^2 c_1/\bar{\tau}^2}{s^2 + \frac{a_1}{\bar{\tau}} s + k_1 c_1/\bar{\tau}^2} + k_1 + \epsilon\right) I_2 & 0 \\
        0 & I_2
    \end{bmatrix}, \quad
    \Psi_2(s) = \begin{bmatrix}
        \left(\frac{\frac{k_2(c_2-a_2)}{\bar{\tau}} s - k_2 b_2/\bar{\tau}^2}{s^2 + \frac{a_2}{\bar{\tau}} s + b_2/\bar{\tau}^2} + k_2 + \delta\right) I_2 & 0 \\
        0 & I_2
    \end{bmatrix}.
\end{aligned}
\end{align}
The resulting IQC-induced LTI system $\{\Psi_k\}_{k=1}^2$ satisfies Assumption~\ref{Ass2}, and its dynamics can be expressed in the form of (\ref{IQCsystem}) with system matrices:
\begin{align*}
\begin{aligned}
    A_\psi &= \begin{bmatrix}
        \tilde{A}_\psi & 0 \\
        0 & \tilde{A}_\psi
    \end{bmatrix}, \quad
    B_{\psi1} = \begin{bmatrix}
        \tilde{B}_{\psi1} & 0 \\
        0 & \tilde{B}_{\psi1}
    \end{bmatrix}, \quad
    B_{\psi2} = 0, \\
    \bar{C}_{\psi,1} &= \begin{bmatrix}
        \tilde{\bar{C}}_{\psi,1} & 0 \\
        0 & \tilde{\bar{C}}_{\psi,1}
    \end{bmatrix}, \quad
    \bar{D}_{\psi1,1} = \begin{bmatrix}
        k_1 + \epsilon & 0 \\
        0 & k_1 + \epsilon
        \end{bmatrix},  \quad
    \bar{D}_{\psi2,1} = 0, \\
    \bar{C}_{\psi,2} &= \begin{bmatrix}
        \tilde{\bar{C}}_{\psi,2} & 0 \\
        0 & \tilde{\bar{C}}_{\psi,2}
    \end{bmatrix}, \quad
    \bar{D}_{\psi1,2} = \begin{bmatrix}
        k_2+\delta & 0 \\
        0 & k_2+\delta
        \end{bmatrix},  \quad
    \bar{D}_{\psi2,2} = 0,
\end{aligned}
\end{align*}
with
\begin{align*}
\tilde{A}_\psi &= \begin{bmatrix}
    0 & 1 & 0 & 0 \\
    -\frac{k_1c_1}{\bar{\tau}^2} & -\frac{a_2}{\bar{\tau}} & 0 & 0 \\
    0 & 0 & 0 & 1 \\
    0 & 0 & -\frac{b_2}{\bar{\tau}^2} & -\frac{a_2}{\bar{\tau}}
\end{bmatrix}, \quad
\tilde{B}_{\psi1} = \begin{bmatrix}
    0 \\
    1 \\
    0 \\
    1
\end{bmatrix}, \\
\tilde{\bar{C}}_{\psi,1} &= \begin{bmatrix}
    -\frac{k_1^2c_1}{\bar{\tau}^2} & \frac{k_1(c_1-a_1)}{\bar{\tau}} & 0 & 0
\end{bmatrix}, \quad
\tilde{\bar{C}}_{\psi,2} = \begin{bmatrix}
    0 & 0 & -\frac{b_2k_2}{\bar{\tau}^2} & \frac{k_2(c_2-a_2)}{\bar{\tau}}
\end{bmatrix}.
\end{align*}

Based on such a system setup, we will first examine the benefits of employing parameter-dependent Lyapunov functions compared with constant Lyapunov functions.
For $r \leq 0.5$, both IQC multipliers in (\ref{Ex.Pi12}) are used for controller synthesis, whereas for $r > 0.5$, only the multiplier $\Pi_2$ is applied.
The parameter space is gridded by 11 points.
Consistent with the observations in \cite{YuaW.TCYB16}, increasing the delay-derivative bound
$r$ leads to a reduction in the admissible delay margin. Nevertheless, compared with existing approaches \cite{WuG.Au01, ZhaG.MCCA05, BriSL.SCL10}, the proposed IQC-based design consistently yields substantially larger delay margins and remains feasible even in regimes where alternative methods fail.

Beyond enhanced delay tolerance, the proposed framework provides a systematic mechanism for handling delay effects by characterizing the delay dynamics exclusively through selected IQC multipliers. This feature offers improved flexibility, allowing the approach to accommodate various delay characteristics, provided appropriate IQCs are available. Furthermore, the IQC-based design exhibits reduced conservatism, leading to improved
${\cal L}_2$ gain performance across a broad range of delay conditions.

Three design methods are considered for comparison:
\begin{enumerate}
\item
Delay-dependent LPV control using a single quadratic Lyapunov funtion;
\item
LFT-based exact-memory control obtained by treating the scheduling parameter $\rho$ as an uncertainty;
\item
Delay-dependent LPV control using parameter-dependent Lyapunov functions with
\begin{align*}
f_1(\rho) &= 1, \quad f_2(\rho) = \rho, \quad f_2(\rho) = \rho^2 \\
g_1(\rho) &= 1, \quad g_2(\rho) = \rho.
\end{align*}
\end{enumerate}
Since the LFT exact-memory approach does not explicitly utilize real-time parameter information, its performance is slightly inferior to that of the LPV-based controllers. In contrast, parameter-dependent Lyapunov functions, which explicitly capture parameter variations with bounded rates, generally provide less conservative results than constant Lyapunov function methods. In particular, for relatively small parameter variation rates, the parameter-dependent formulation yields noticeable improvements in controlled performance, as clearly illustrated in Table \ref{Tab.gamma}.

\begin{table}[htb]
\centering
\caption{Performance comparison of different delay dependent control methods}
\label{Tab.gamma}
\begin{tabular}{c|p{0.6in}|cccccc} \hline \hline
Method & parameter variation rate $\nu$ & \multicolumn{6}{c}{(Delay derivative $r$, delay bound $\tau$)}  \\ \hline
& & (0, 10) & (0.5, 2.5) & (0.9, 1) & (1.5, 1) & (1.7, 2.5) \\ \hline
Quadratic LF & & 3.6859 & 3.6494 & 1.8506 & 1.9381 & 3.9573 \\
LFT approach \cite{YuaW.TCYB16} & & 3.6864 & 3.6507 & 1.8510 & 1.9383 & 3.9580 \\ \hline
\multirow{5}{1.0in}{Parameter-dependent LF} & 0.1 & 3.4501 & 2.4843 & 1.8472 & 1.8982 & 2.5334 \\
& 0.5 & 3.5783 & 3.1911 & 1.8497 & 1.9287 & 3.3433 \\
& 1 & 3.5903 & 3.2886 & 1.8500 & 1.9311 & 3.4650 \\
& 5 & 3.6068 & 3.4625 & 1.8501 & 1.9337 & 3.6888 \\
& 10 & 3.6117 & 3.5338 & 1.8502 & 1.9344 & 3.7879 \\ \hline \hline
\end{tabular}
\end{table}

The improved performance obtained with parameter-dependent Lyapunov functions can be attributed to their ability to explicitly capture the dependence of the system dynamics on the scheduling parameter. Unlike constant Lyapunov functions, which enforce a uniform stability certificate over the entire parameter set, parameter-dependent Lyapunov functions allow the Lyapunov matrix to vary with $\rho$.
This additional degree of freedom enables a tighter characterization of the system’s stability and performance properties, thereby reducing conservatism.
%
Consequently, the parameter-dependent Lyapunov framework typically results in less conservative stability and performance conditions, explaining the observed improvement in the achievable ${\cal L}_2$ gain levels.
On the other hand, increasing either the maximum delay bound $\bar{\tau}$ or the delay derivative bound $r$ leads to larger values of $\gamma$, indicating degraded performance. 

To illustrate the proposed delay-dependent control strategy via time-domain simulations, we consider the time-varying parameter $\rho(t) = \sin (0.5 t)$ and the state-delay function
$\tau(t) = 0.2 \sin(6 t) + 1.8$, yielding the bounds
$(r, \bar{\tau}) = (1.2, 2)$. Since the delay-derivative bound exceeds unity, conventional methods \cite{WuG.Au01, ZhaG.MCCA05, BriSL.SCL10} are not applicable. In contrast, by solving the optimization problem (\ref{Opt}) with the IQC multiplier $\Pi_2$ and parameter variation bound $\nu = 1.0$, we obtain an optimized ${\cal L}_2$ gain
$\gamma = 2.1137$
along with the corresponding controller gains.
With zero initial conditions and a unit pulse disturbance applied over $[0,2] sec$,
simulations are performed using the synthesized controller. The system responses are depicted in Fig. \ref{Fig.Ex}.
Despite the presence of parameter variations and time-varying state delays, the plant states converge rapidly to zero while maintaining a reasonable control effort.

\begin{figure}[htb!]
\centering
\subfloat[plant states: $x_1$ (solid line), $x_2$ (dash line)]{
\includegraphics[width=0.49\textwidth]{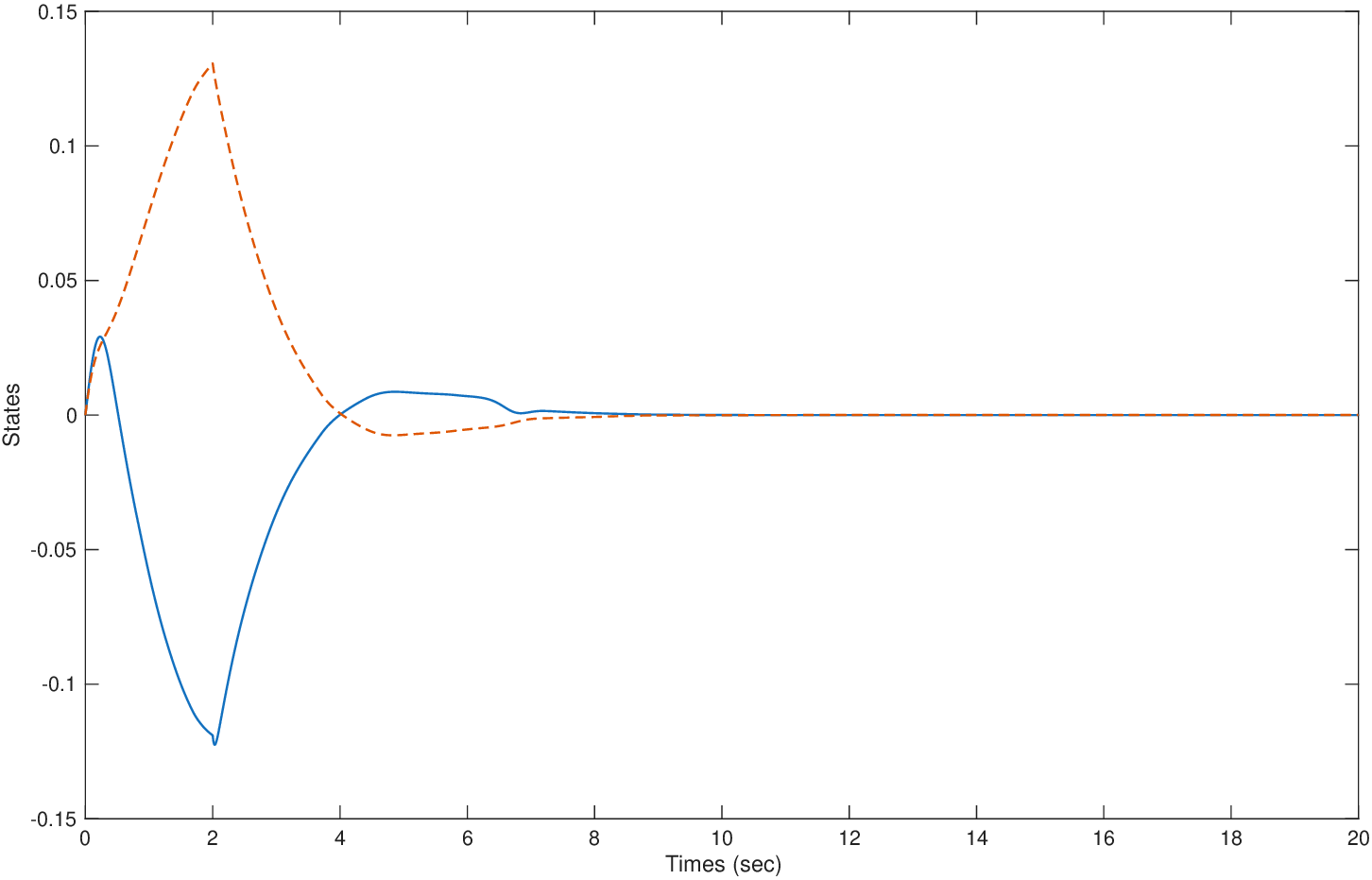}
}
\subfloat[delayed states]{
\includegraphics[width=0.49\textwidth]{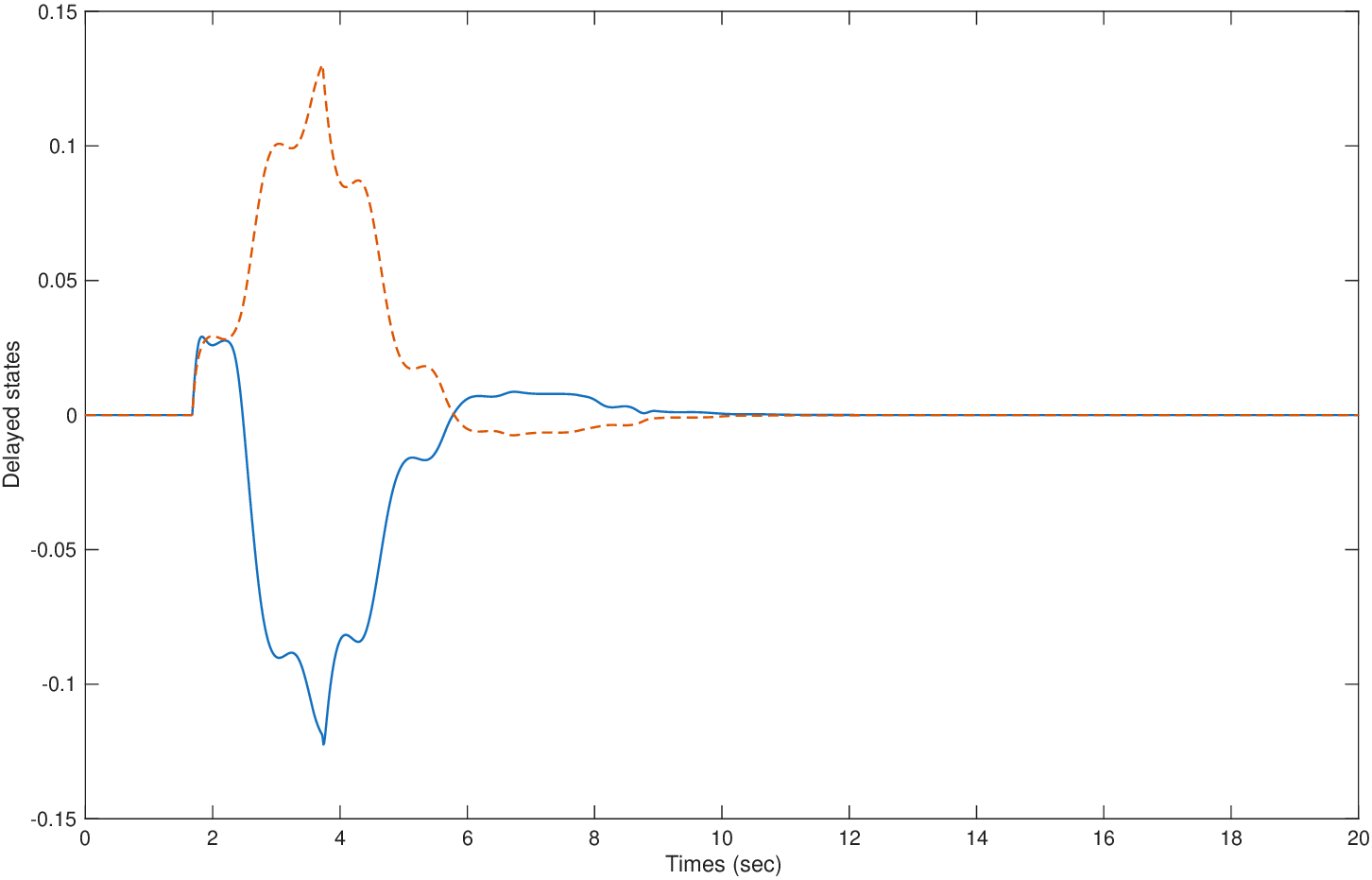}
} \\
\subfloat[control input]{
\includegraphics[width=0.49\textwidth]{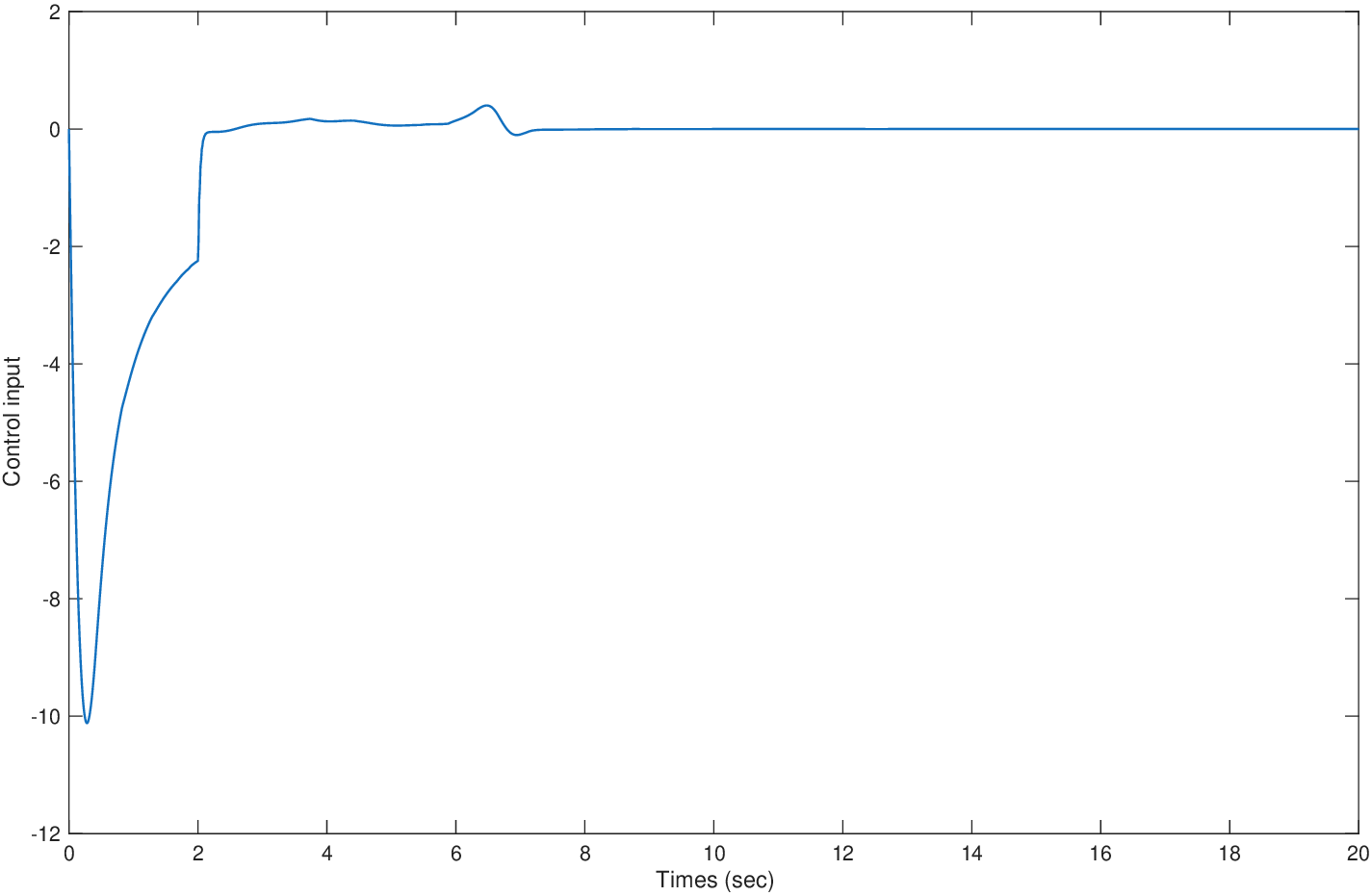}
}
\caption{Simulation results.}
\label{Fig.Ex}
\end{figure}


\section{Concluding Remarks}
\label{Sec.Con}

In this paper, a novel delay-dependent control framework has been developed within the IQC-based robust analysis and synthesis paradigm for LPV systems with time-varying state delays. The proposed controller adopts a state-feedback structure composed of a conventional memoryless component augmented with an additional term that explicitly captures delay-dependent dynamics. Dynamic IQCs are employed to characterize the input–output behavior of the state-delay operator, enabling a unified treatment of robust stability analysis and controller synthesis. Delay-dependent LPV synthesis conditions guaranteeing robust $\mathcal{L}_2$-gain performance are derived in terms of parameter-dependent LMIs, which can be solved efficiently via convex optimization. Numerical studies demonstrate that the proposed delay-dependent control strategy achieves improved closed-loop performance while maintaining a systematic and computationally tractable design procedure.

From a methodological perspective, the paper introduces a delay-dependent state-feedback controller with exact memory for LPV systems with time-varying delays and establishes a unified IQC-based synthesis framework using dynamic multipliers and parameter-dependent Lyapunov functions. Compared with classical delay-dependent approaches based on Lyapunov–Krasovskii functionals, the proposed formulation provides a less conservative and more structured design procedure while retaining favorable computational properties. The developed IQC-based approach offers a simple yet rigorous framework for delay-dependent controller design. Owing to the generality of the IQC formulation, the results can be naturally extended to accommodate alternative performance objectives, such as robust $\mathcal{H}_2$ performance, as well as broader classes of uncertainties and nonlinearities, including parametric uncertainty, unmodeled dynamics, saturation, deadzone, and slope-bounded nonlinearities. An important direction for future work is the application of the proposed methodology to networked control systems, where communication delays and uncertainties are intrinsic features.

\end{document}